\documentclass[aps,pra,twocolumn,nofootinbib,floatfix,longbibliography,superscriptaddress]{revtex4-2}

\usepackage{amsmath,amssymb,amsthm,bbm,bm,color,dsfont,float,graphicx,hyperref,makecell,mathrsfs,mathtools,nicefrac,pgfplots,physics,tikz,times,txfonts}
\usepackage{comment}
\usepackage[normalem]{ulem}  

\definecolor{equationcolor}{RGB}{222,94,100}
\definecolor{alecolor}{RGB}{238,33,80}

\hypersetup{urlcolor=equationcolor, colorlinks=true,citecolor=equationcolor,linkcolor=equationcolor} 
\pgfplotsset{compat=1.18} 

\DeclareFontFamily{U}{mathb}{\hyphenchar\font45}
\DeclareFontShape{U}{mathb}{m}{n}{
	<-6> mathb5 <6-7> mathb6 <7-8> mathb7
	<8-9> mathb8 <9-10> mathb9
	<10-12> mathb10 <12-> mathb12
}{}
\DeclareSymbolFont{mathb}{U}{mathb}{m}{n}
\DeclareMathSymbol{\ggcurly}{\mathrel}{mathb}{"CF}

\renewcommand{\v}[1]{\ensuremath{\boldsymbol #1}}

\makeatletter
\def\blfootnote{\gdef\@thefnmark{}\@footnotetext}
\makeatother

\theoremstyle{plain}
\newtheorem{thm}{Theorem}

\newtheorem{lem}[thm]{Lemma}
\newtheorem{cor}[thm]{Corollary}

\def\>{\rangle}
\def\<{\langle}

\newlength\myindent
\begin{document}
\newcommand{\lanl}{Theoretical Division (T4), Los Alamos National Laboratory, Los Alamos, New Mexico 87545, USA.}
\newcommand{\rochestar}{Department of Physics and Astronomy, University of Rochester, Rochester, New York 14627, USA}
\newcommand{\iitj}{Department of Physics, Indian Institute of Technology Jodhpur, Jodhpur 342030, India}
\author{Tanmoy Biswas}
\email{tanmoy.biswas23@lanl.gov}
\affiliation{\lanl}
\author{Chandan Datta}
\affiliation{\iitj}
\author{Luis Pedro Garc\'ia-Pintos}
\affiliation{\lanl}

\title{
Quantum thermodynamic advantage in work extraction from steerable quantum correlations
}

\date{\today}
\begin{abstract}
     
Inspired by the primary goal of quantum thermodynamics---to characterize quantum signatures and leverage their benefits in thermodynamic scenarios---, we design a work extraction task within a bipartite framework that exhibits a quantum thermodynamic advantage. The steerability of quantum correlations between the two parties is the key resource enabling such an advantage. In designing the task, we exploit the correspondence between steerability and the incompatibility of observables. Our work extraction protocol involves mutually unbiased bases, which exhibit maximum incompatibility and therefore maximum steerability, showcased in maximally entangled quantum states. The work extraction protocols consist of quenches and thermalization processes, which serve as fundamental building blocks for various thermodynamic protocols, such as heat engines.
We derive upper bounds on the extractable work for unsteerable and steerable correlations and devise a protocol that saturates the latter. The ratio between the extractable work in steerable and unsteerable scenarios, which encapsulates the quantum advantage, increases with the dimension of the underlying system (sometimes referred to as an unbounded advantage). This proves a quantum thermodynamic advantage arising from steerable quantum correlations.
\end{abstract}
\maketitle
\section{Introduction}

Exploiting quantum correlations between spatially separated systems is a cornerstone of the second quantum revolution \cite{dowling2003quantum, Deutsch_PRXQ,Aspect2023}. 
Characterizing these correlations not only deepens our understanding of the fundamental nature of the universe, but also paves the way for the development of novel technologies that can significantly outperform their classical counterparts~\cite{Bennett_teleportation,Ekert_QKD,Bouwmeester1997,Bao2012,ExptEntangled,Gröblacher2007}. Quantum steering refers to quantum correlations in which a spatially separated observer can influence the set of conditional quantum states accessible to a distant party through uncharacterized measurements in a way that cannot be explained by local causal influences~\cite{Wiseman_steering, SteeringRPP, SteeringRevModPhys}. (``Uncharacterized measurements'' assume that the party is unaware of the specific measurement operator being applied and, therefore, does not have access to the corresponding conditional states.).  
When both parties use uncharacterized measurements, this scenario aligns with the device-independent framework of a standard Bell test~\cite{Brunner_nonlocality_review, ScaraniBook, ScaraniOutlook, Acin2016, MassarPironio}. Quantum steering exploits fundamentally quantum correlations with no classical counterpart and is recognized as part of a broader family of non-classical correlations.  Quantum correlations in entangled states that violate Bell inequalities also exhibit steerability; however, the converse does not hold. 
Steerable states lie between entangled ones and states that can be used to violate Bell inequalities: all states that violate a Bell inequality are steerable, and all steerable states are entangled (but, not all entangled states violate a Bell inequality).  

The concept of quantum steering was first introduced by Schrodinger in response to the Einstein-Podolsky-Rosen (EPR) paradox~\cite{EPR_paradox,Schrodinger_1935,Schrodinger_1936}. The modern formulation of quantum steering within the framework of quantum information theory, as presented in Ref.~\cite{Wiseman_steering}, serves as a method to verify entanglement when an agent Alice uses uncharacterized measurement devices. Quantum steering provides significant advantages in various information-theoretic tasks, including quantum key distribution \cite{Branciard_QKD,Ma2012}, quantum key authentication \cite{Mondal_QKA}, randomness certification \cite{Law_2014, Passaro_2015}, sub-channel discrimination \cite{Piani2009,Piani_subchannel}, distinguishing quantum measurements \cite{Datta_2021}, characterization of quantum measurement incompatibility \cite{Chen_incompatibility, Cavalcanti_incompatibility}, secret sharing, \cite{Xiang_secret_sharing,Kogias2015}, quantum metrology \cite{Yadin_metrology}, and quantum teleportation \cite{Reid2013, He2015,Pirandola2015, Cavalcanti_teleportation}. 


In this paper, we prove that quantum steering serves as a valuable resource in quantum thermodynamics, enabling a quantum advantage in work extraction that grows with the dimension of the underlying system. Building on the central objective of quantum thermodynamics, exploring and characterizing quantum signatures within thermodynamic contexts~\cite{OppHorodecki3,Francica_ergotropy_coherence,Francica_ergotropy_discord,Francica_correlations_ergotropy,Lostaglio_LR_PRL,Steering_Engines,Steering_Engines_expt, Touil_2022,hsieh2024generalquantumresourcesprovide,Biswas2022extractionof,BiswasPRL,BiswasPRE,Hsieh_PRL}, we introduce a work extraction task in a bipartite setting. 
Inspired by Maxwell's demon~\cite{maxwell1871theory,szilard1929entropy,Maxwell_demon_review}, Landauer's principle~\cite{Landauers_principle,Landauers_principle2,Sagawa1,Sagawa2}, and recent work that explores measurement incompatibility in a thermodynamic context~\cite{Hsieh_PRL}, we design a work extraction task where Bob seeks to extract work from his quantum system, guided by measurement and feedback from a distant party, Alice. 
Our findings reveal that substantially more work can be extracted when both systems share steerable correlations (e.g., in a maximally entangled state), compared to scenarios where such steerability is absent.

We prove that the ratio of extracted work, with and without steerable correlations, scales as \(\mathcal{O}(\sqrt{d})\), when the dimension \(d\) of the underlying quantum system is an integer power of a prime number~\footnote{This constraint on dimensionality arises from the existence of $d+1$ mutually unbiased bases for a dimension equals to integer power of a prime number.}. This result implies a quantum thermodynamic advantage in work extraction. 
To the best of our knowledge, this is the first proof of a quantum advantage in the extractable work in a thermodynamic context where the advantage grows with the dimension of the underlying system.
Such a quantum thermodynamic advantage also enables us to interpret the work extracted from our proposed task as a witness of steerability, which is sufficient to certify the entanglement of the shared quantum state. 


\begin{figure*}[htbp]
    \centering
    \includegraphics[width=17.5cm]{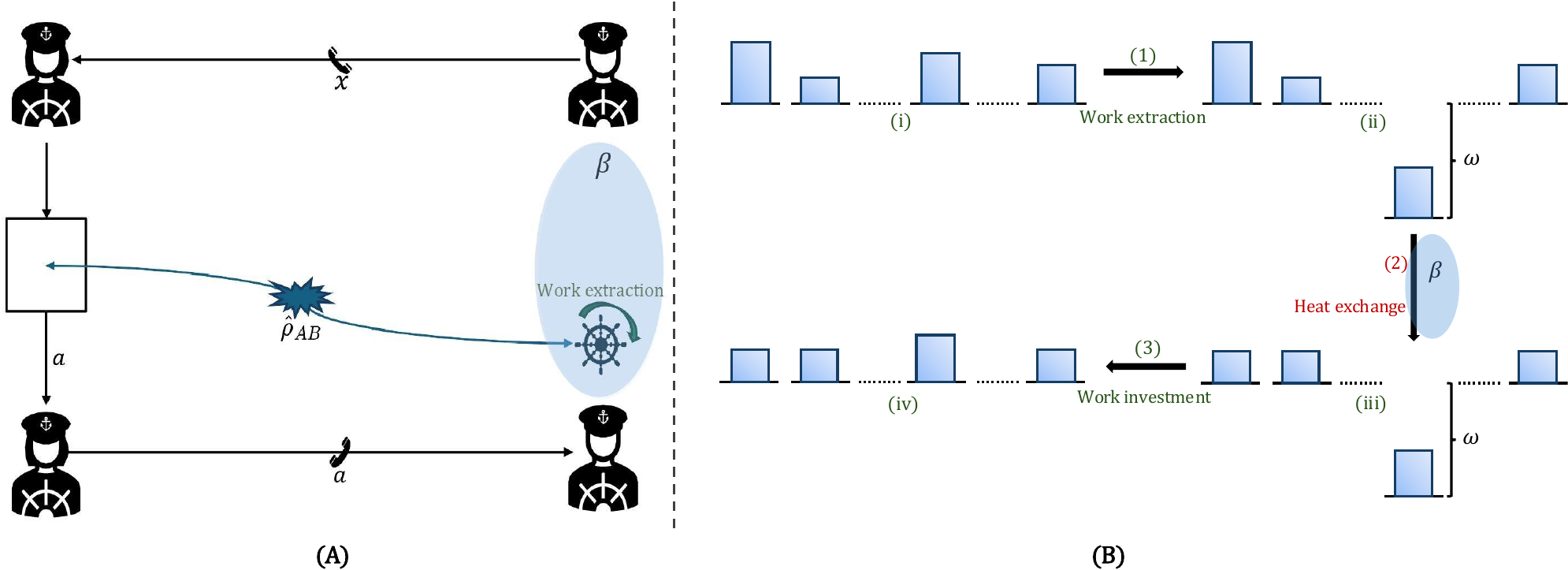}
    \caption{\label{fig:Steering_task} Panel (A) illustrates the work extraction task. Bob selects a random input \( x \) uniformly from \( \mathcal{X} \) and communicates it to Alice, who then performs a measurement on her half of the system and sends the outcome \( a \) back to Bob. Based on \( a \) and \( x \), Bob executes the work extraction protocol shown in Panel (B). Due to Alice's measurement, Bob obtains the state \(\hat{\rho}_{a|x}\), initially associated with a trivial Hamiltonian, as depicted in stage (i). The solid lines represent the energy levels of the Hamiltonian at each stage, while the rectangles indicate the population associated with these energy levels. To extract work, Bob applies a quench that changes the Hamiltonian to \( H_{a|x} \), as shown in stage (ii), without interacting with the heat bath (In the picture above, \( H_{a|x} \) resembles the Hamiltonian given in Eq. \eqref{MUB_Hamiltonian_main}, though, in general, it can represent any non-trivial Hamiltonian.). He then thermalizes the system by coupling it to the heat bath, transforming \(\hat{\rho}_{a|x}\) into the thermal state \(\hat{\gamma}_{a|x}\), as illustrated in stage (iii), which involves heat exchange. Finally, Bob restores the Hamiltonian to its initial value by applying another quench, requiring work investment, as depicted in stage (iv). On average, the extractable work is $|\mathcal{X}|^{-1} \sum_{a, x} p(a|x) \left[-\Tr(\hat{\rho}_{a|x} H_{a|x}) + \Tr(\hat{\gamma}_{a|x} H_{a|x})\right]$. The extractable work differs if Alice and Bob's systems share steerable quantum correlations or not.}
\end{figure*}
\section{Main results}

\subsection{Work extraction task} 

We consider a $d$-dimensional system, representing the working body, with a trivial Hamiltonian ($H=0$). Bob, assisted by Alice, aims to extract work from this quantum system. Bob has access to a fixed set of Hamiltonians to which he can quench, as well as a heat bath at temperature \( T \). We denote the initial state of the combined system shared between Alice and Bob by \( \hat{\rho}_{AB} \). The initial marginal state of Bob's system is $\Tr_{A}\left(\hat{\rho}_{AB}\right)$.  The work extraction task, illustrated in Fig. \ref{fig:Steering_task}, goes as follows:  
\begin{enumerate}
    \item Bob samples an input \( x \) uniformly from the set \( \mathcal{X} \) and transmits it to Alice through a classical communication channel.
    \item  Upon receiving \( x \), Alice performs a measurement on her subsystem, described by a positive operator-valued measure (POVM) \( M_x^a \). She communicates the measurement outcome \( a \) to Bob. Performing a measurement on Alice's subsystem updates the state of Bob's subsystem to \( \hat{\rho}_{a|x} \), given by  
   \begin{eqnarray}\label{assemblage}
     &&\hat{\rho}_{a|x} = \frac{\rho_{a|x}}{p(a|x)}, \quad \text{where} \quad\rho_{a|x}=\Tr_A\left[(M_x^a \otimes \mathbb{I}_B) \hat{\rho}_{AB}\right].
     \nonumber\\&&\text{and}\quad p(a|x) = \Tr\left[(M_x^a \otimes \mathbb{I}_B) \hat{\rho}_{AB}\right].
   \end{eqnarray}
   Here $p(a|x)$ represents the conditional probability of obtaining the outcome $a$ given the input $x$. We denote normalized states by $\hat \rho$. 
   
   \item  Based on the input $x$ and the outcome $a$, Bob proceeds as follows to extract work:
   
   \begin{enumerate}
       \item Quench the initial Hamiltonian \( H = 0 \) to \( H_{a|x} \) [see Eq.~\eqref{MUB_Hamiltonian_main} below]. The system remains isolated from the heat bath, so its state $\rho_{a|x}$ remains unchanged. The difference between the initial and final average energies, \( -\Tr\left(H_{a|x} \hat{\rho}_{a|x} \right) \), can be extracted as work \cite{Aberg2013,Entropy_sagawa}. Next, we connect the heat bath to thermalize the state \( \rho_{a|x} \) with respect to the Hamiltonian \( H_{a|x} \). This results in the following thermal state \( \hat{\gamma}_{a|x} \) for Bob's system:
      \begin{equation}
        \hat{\gamma}_{a|x} = \frac{e^{-\beta H_{a|x}}}{\Tr\left(e^{-\beta H_{a|x}}\right)},\quad\text{with} \quad \beta = \left(k_B T\right)^{-1},
      \end{equation}
where $k_B$ is the Boltzmann constant. No work extraction occurs during such a thermalization process \cite{Aberg2013,Horodecki2013}.

       \item Quench the Hamiltonian from \( H_{a|x} \) back to the initial Hamiltonian \( H = 0 \), without interacting with the heat bath. The work invested is given by the difference between the initial and final average energy during this process, \( \Tr(H_{a|x} \hat\gamma_{a|x}) \). The Hamiltonian returns to its initial value, so the process is cyclic~\cite{Ergotropy}.
   \end{enumerate}
   \end{enumerate}
   
Therefore, when Bob samples the input $x$ and Alice obtains the outcome $a$, the net extractable work is 
   \begin{equation}\label{diff_of_FE}
      W\left(\hat{\rho}_{a|x}, H_{a|x}\right) \coloneqq -\Tr\left(H_{a|x} \hat{\rho}_{a|x} \right) + \Tr(H_{a|x} \hat\gamma_{a|x}).
   \end{equation} 
The extracted work averaged over all possible inputs $x$ and outcomes $a$ is  
\begin{eqnarray}
    &&W\left(\{\hat\rho_{a|x}\},\{H_{a|x}\}\right) \coloneqq \frac{1}{|\mathcal{X}|} \sum_{a, x} p(a|x) W(\hat{\rho}_{a|x}, H_{a|x})\nonumber\\ &=&  \frac{1}{|\mathcal{X}|}\sum_{a, x} p(a|x) \left[-\Tr\left(H_{a|x} \hat{\rho}_{a|x} \right) + \Tr(H_{a|x} \hat\gamma_{a|x})\right],
\end{eqnarray}
where the prefactor $1/|\mathcal{X}|$ arises due to the uniform sampling of the input from the set $\mathcal{X}$. Here $|\mathcal{X}|$ denotes the number of elements in the set $\mathcal{X}$. Our work extraction protocol in step 3 resembles the protocol used in the four-stroke microscopic Otto engine to extract work~\cite{cangemni_Levy_engines,Equivelence_Uzdin_Kosloff,Blickle2011,Singer}.

Drawing insights from the correspondence between unsteerability and joint measurability~\cite{Uolaonetoone,Ku2022}, we focus on the work extraction task introduced above. We assume that Bob can quench his system to Hamiltonians
\begin{equation}\label{MUB_Hamiltonian_main}
    H_{a|x}=-\omega \ketbra{\phi^a_x},
\end{equation}
where $\ket{\phi^a_x}$ is the $a^{\text{th}}$ element of mutually unbiased bases (MUB) indexed by $x$. $d$ is the dimension of the quantum system $s$, and $a\in\{0,\ldots,d-1\}$ and $x\in\{0,\ldots,n-1\}$ (thus, $|\mathcal{X}|=n$ in this scenario).
Recall that a MUB is defined via the following relation:
\begin{equation}
    \left|\braket{\phi_{x}^a}{\phi_{y}^b}\right|=\begin{cases}
			\delta_{a,b}, & \text{if $x=y$ }\\
            \frac{1}{\sqrt{d}}, & \text{otherwise}
		 \end{cases}
\end{equation}
Hamiltonians similar to Eq.~\eqref{MUB_Hamiltonian_main} are often considered in analog quantum search algorithms \cite{Farhi1998,Roland_Cerf,Allahverdyan2022}. 

Figure~\ref{fig:Steering_task} illustrates the work extraction protocol.
Next, we determine the maximum amount of work that can be extracted when the correlations in Alice and Bob's underlying quantum state are unsteerable. (We summarize the definition of steerable and unsteerable correlations in  Sec.~I of the Supplementary Material.) 

\subsection{Extractable work for unsteerable correlations and a steerability witness}

When correlations between the quantum states of Alice's and Bob's systems are unsteerable, the post-measurement state $\hat{\rho}_{a|x}$ can be expressed as~\cite{Wiseman_steering}
\begin{equation}\label{LHS_decomposition_main}
    \hat{\rho}_{a|x} = \int d\lambda \; p(\lambda|a,x) \; \hat\rho_{\lambda};  \quad\forall \; a,\; x.
\end{equation}
The decomposition in Eq. \eqref{LHS_decomposition_main} is referred as local hidden state (LHS) decomposition. Combining Eqs.~\eqref{assemblage} and~\eqref{LHS_decomposition_main} imply
\begin{equation}\label{LHS_decomposition_assemblage_main}
    \rho_{a|x}=  \int d\lambda\; p(\lambda)\; p(a|x,\lambda)\;\hat\rho_{\lambda}.\quad  \forall \; a,\;x,
\end{equation}
often referred to as an unsteerable assemblage~\cite{Wiseman_steering, SteeringRevModPhys}.  We denote the set of all unsteerable assemblages by $\mathcal{L}$  . A brief description of the unsteerability of correlations is provided in Sec.~I of   of the Supplementary Material.
The following theorem upper bounds the amount of work that can be extracted in the unsteerable scenario:
\begin{thm}\label{LHS_theorem_upper_bound_work_extraction}
    Consider the work extraction task where Bob can quench to the Hamiltonians \( H_{a|x} \) in Eq. \eqref{MUB_Hamiltonian_main}, and the state \( \hat{\rho}_{a|x} \) admits the LHS decomposition described in Eq.~\eqref{LHS_decomposition_main}. Then, the average amount of extracted work, denoted by $W_{LHS}$, satisfies 
    \begin{equation}\label{WLHS_Bound}
        W_\mathrm{LHS} \leq  \frac{\omega}{d}\left(1+\frac{d-1}{\sqrt{n}}\right)-\omega\left(\frac{e^{\beta\omega}}{e^{\beta\omega}+d-1}\right) \eqqcolon W_{\mathrm{classical}}.
    \end{equation}
\end{thm}
The second term on the right-hand side of the inequality in Eq. \eqref{WLHS_Bound} is the average energy of the Gibbs state with respect to the Hamiltonian in Eq.~\eqref{MUB_Hamiltonian_main}. We prove Theorem~\ref{LHS_theorem_upper_bound_work_extraction} in Appendix~\ref{Appendix_proof}. In this paper, we consider the unsteerability of correlations as a signature of classicality, while identifying the steerability of correlations as a signature of nonclassicality. Therefore, we denote the upper bound on the amount of work extracted in the unsteerable scenario as \( W_{\text{classical}} \). The inequality in Eq.~\eqref{WLHS_Bound} can serve as a witness of steerability in terms of the extracted work which we state formally in the following corollary: 

\begin{cor}
If the work \( W \) extracted from a protocol exceeds \( W_{\mathrm{classical}} \), the underlying quantum state \(\hat{\rho}_{AB}\) shared between Alice and Bob possesses steerable correlations.
\end{cor}

If the dimension $d$ is an integer power of a prime number, there exist $d+1$ MUBs~\cite{Ivonovic_1981,WOOTTERS1989363,Bandyopadhyay2002}. Thus, we can take $n=d+1$. Then, Eq.~\eqref{WLHS_Bound} implies
\begin{equation}\label{W_classic_main}
    W_\text{classical}= \frac{\omega}{d} \left(1+\frac{d-1}{\sqrt{d+1}}\right)-\omega\left(\frac{e^{\beta\omega}}{e^{\beta\omega}+d-1}\right).
\end{equation}
Thus, the extracted work can increase at most as $\mathcal{O}\left(1/\sqrt{d}\right)$ when correlations are unsteerable. Next, we focus on the quantum scenario, when the combined state \(\hat{\rho}_{AB}\) contains steerable correlations.

\subsection{Extractable work in the quantum scenario}
In this section, we explore the average amount of work as defined in Eq. \eqref{Avg_amount_of_extracted_work}, in the quantum scenario, where \(\hat{\rho}_{a|x}\) does not necessarily admit an LHS decomposition as in Eq. \eqref{LHS_decomposition_main}.
\begin{thm}\label{Quantum_theorem_upper_bound_work_extraction}
    Consider the work extraction task in which Bob can quench to the Hamiltonians \( H_{a|x} \) in Eq. \eqref{MUB_Hamiltonian_main}. Then, the 
     maximum achievable average extracted work is given by   
\begin{equation}\label{Wquantum_Bound}  
    W_{\mathrm{quantum}} \coloneqq \omega - \omega \left(\frac{e^{\beta\omega}}{e^{\beta\omega} + d - 1}\right).  
\end{equation}
\end{thm}
\begin{proof}
For any state, the average extractable work in Eq.~\eqref{Avg_amount_of_extracted_work} is upper bounded by 
    \begin{equation}
    \label{upper_bound_main_work_main}
        W(\{\hat\rho_{a|x}\},\{H_{a|x}\}) \leq  \frac{1}{|\mathcal{X}|}\sum_{a, x} p(a|x) \left[-\sigma_{\text{min}}\left(H_{a|x}\right) + \Tr(H_{a|x} \hat\gamma_{a|x})\right],
    \end{equation}
where $\sigma_{\text{min}}(\cdot)$ denotes the smallest eigenvalue of $(\cdot)$. For $H_{a|x}$ in Eq.~\eqref{MUB_Hamiltonian_main}, $\sigma_{\text{min}}\left(H_{a|x}\right)=-\omega$, so  Eq.~\eqref{upper_bound_main_work_main} reduces to
\begin{eqnarray}\label{generic_upper_bound_optimality}
    W(\{\hat\rho_{a|x}\},\{H_{a|x}\})\leq \omega-\omega\left(\frac{e^{\beta\omega}}{e^{\beta\omega}+d-1}\right).
\end{eqnarray}

Next, we show a protocol that attains such a bound. 
Assume the initial state $\hat{\rho}_{AB}$ of the combined system of Alice and Bob is a maximally entangled state of the form
\begin{equation}
    \hat{\rho}_{AB}=\frac{1}{d}\sum_{i=0}^{d-1}\sum_{j=0}^{d-1}\ketbra{i}{j}\otimes\ketbra{i}{j}.
\end{equation}
(The correlation present in the maximally entangled state is known to be steerable, as the marginal state \(\hat{\rho}_{a|x}\) cannot always be expressed in the form given by Eq.~\eqref{LHS_decomposition_main} \cite{SteeringRevModPhys}.) Suppose that Bob samples the input $x$ randomly from the set $\mathcal{X}$ in a particular run. After Bob communicates the input $x$ to Alice, she performs a projective measurement on the orthonormal basis
\begin{equation}\label{MUB_measurement_assemblage}
    \left\{\ketbra{\phi^{*c}_{x}}\right\}_{c=0}^{d-1}, \qquad \textnormal{}
    \ket{\phi^{*c}_{x}}=\sum_{i=0}^{d-1}\braket{\phi^c_x}{i}\ket{i},
\end{equation}
obtains an outcome $a$ from the set $\{0,\ldots,d-1\}$, and communicates the outcome to Bob. 
(The set of projective measurements indexed by $x$ forms a MUB). 
This results in the following state on Bob's side:
\begin{eqnarray}\label{MUB_assemblage_main}
    \hat\rho_{a|x}=\ketbra{\phi^a_x}.
\end{eqnarray}
We derive Eq.~\eqref{MUB_assemblage_main} in Appendix~\ref{Derivation_eq}. Substituting $\hat\rho_{a|x}$ from Eq. \eqref{MUB_assemblage_main} and $H_{a|x}=-\omega \ketbra{\phi^a_x}$ in the extractable work in Eq. \eqref{diff_of_FE}, we obtain that the extractable work $W_{\text{quantum}}$ in the quantum scenario is
\begin{equation}\label{W_quantum_main}
    W_{\text{quantum}}:=\omega-\omega\left(\frac{e^{\beta\omega}}{e^{\beta\omega}+d-1}\right).
\end{equation}
\end{proof}

\subsection{Quantum thermodynamic advantage in work extraction}

Using Theorems~\ref{LHS_theorem_upper_bound_work_extraction} and~\ref{Quantum_theorem_upper_bound_work_extraction}, the ratio between extractable works in the quantum and unsteerable scenarios satisfies  
\begin{equation}\label{definition}
    \frac{W_{\text{quantum}}}{W_{\text{LHS}}}\geq \frac{W_{\text{quantum}}}{W_{\text{classical}}} = \frac{\omega-\omega\left(\frac{e^{\beta\omega}}{e^{\beta\omega}+d-1}\right)}{\frac{\omega}{d} \left(1+\frac{d-1}{\sqrt{n}}\right)-\omega\left(\frac{e^{\beta\omega}}{e^{\beta\omega}+d-1}\right)}:=\xi.
\end{equation}
The ratio $\xi \coloneqq W_{\text{quantum}}/W_{\text{classical}}$ quantifies the quantum thermodynamic advantage in the work extraction task. A quantum advantage exists for $\xi > 1$, when 
$W_{\text{quantum}} > W_{\text{classical}}\geq W_{\text{LHS}}$. From Eq.~\eqref{definition}, a quantum thermodynamic advantage occurs when
\begin{equation}\label{omegadsqrtn}
    \omega>\frac{\omega}{d}\left(1+\frac{d-1}{\sqrt{n}}\right)\;\;\Longleftrightarrow \frac{d\sqrt{n}}{\sqrt{n}+d-1}>1.
\end{equation}
The above inequality is satisfied for any $d$ if $n \geq 2$,
so two mutually unbiased bases are sufficient for quantum thermodynamic advantage when constructing $H_{a|x}$ as in Eq.~\eqref{MUB_Hamiltonian_main}. Note that at least two mutually unbiased bases exist in any dimension, as one can construct them using any orthonormal basis and its Fourier transform. 
This implies that our protocol exhibits a quantum thermodynamic advantage for any dimension. 

If the dimension \( d \) of the underlying system is an integer power of a prime number, there exist $d+1$ MUBs, so one can take $n=d+1$~\cite{Ivonovic_1981,WOOTTERS1989363,Bandyopadhyay2002}. Then, Eq.~\eqref{definition} implies that $\xi \sim \mathcal{O}(\sqrt{d})$, with $W_{\text{quantum}}\sim \mathcal{O}(1)$ and $W_{\text{classical}}\sim \mathcal{O}(1/\sqrt{d})$, respectively, from Eqs.~\eqref{W_classic_main} and~\eqref{W_quantum_main}.
That is, the quantum scenario offers an advantage in work extraction over the unsteerable scenario that grows with the dimension of the underlying system (sometimes referred to as an unbounded quantum advantage in the information theory literature~\cite{Skrzypczyk,Garner_2017,Heinsaari,SahaPRR, Rout}).

  The scaling advantage described in the previous paragraph depends on the specific set of Hamiltonians used for the work extraction task. In Sec.~II of the Supplementary Material, we consider a different set of Hamiltonians composed of tensor-product Pauli matrices acting on an $n$-qubit system~\cite{SM}. In this setting, we demonstrate a work extraction advantage that scales as $W_{\text{quantum}}/W_{\text{classical}} \geq \mathcal{O}(\sqrt{n})$, which corresponds to a scaling of $\mathcal{O}(\sqrt{\log_2 d})$ with the underlying system dimension $d = 2^n$. It would be interesting to characterize classes of Hamiltonians that can yield scaling advantages in work extraction from quantum correlations.

\section{Discussion}

We showed that the steerability of quantum correlations can be a valuable resource in thermodynamic scenarios. To illustrate this, we designed two work extraction tasks in a bipartite setting, where steerable correlations (specifically, in maximally entangled states) between two parties are exploited to demonstrate a quantum thermodynamic advantage in work extraction.   
 In the main work extraction task detailed in this letter, the ratio of extractable work enabled by steerable quantum correlations—arising from the maximally entangled state—to that achievable from any unsteerable classical correlations scales as $\mathcal{O}(\sqrt{d})$. 

 Our research paves the way to several promising directions in quantum thermodynamics. For example, work extraction plays a central role in cooling protocols by enabling heat transfer from a colder to a hotter reservoir, and it is also fundamental to the operation of quantum batteries, where work is used to charge a quantized energy system. It would be interesting to investigate whether our approach to exploit quantum correlations to enhance work extraction can be leveraged to design efficient protocols for cooling or charging. 

Another natural and important direction for future research is to pursue experimental realizations of the quantum thermodynamic advantage predicted by our framework, using platforms such as trapped ions~\cite{Singer} or superconducting qubits~\cite{Aamir2025}. We emphasize that our work extraction protocol relies on quenching and thermalization—two experimentally feasible operations that are routinely implemented in microscopic thermal machines~\cite{cangemni_Levy_engines,PRXMH,Blickle2011,Singer}. Given this, we expect that the predicted work extraction advantage can be demonstrated in near-term experimental setups. 

Additionally, it would be interesting to investigate whether a similar framework could be developed to extract work in closed systems without the need for a heat bath. In the closed system scenario, work extraction occurs by activating an interaction Hamiltonian for a specific duration. The interaction Hamiltonian induces unitary evolution. Thus, the extracted work corresponds to the difference between the average energies of the initial and final pure states. 

 Finally, it would be intriguing to investigate whether such a work extraction task can reveal advantages arising from other forms of quantum correlations. For example, one might ask whether the correlations present in a given entangled state offer an advantage over those found in any separable state. Moreover, it would be particularly compelling to explore potential connections between this work extraction advantage and various quantifiers of entanglement and steerability. Notably, measures such as entanglement and steerability robustness have established operational interpretations in subchannel distinguishability tasks \cite{Piani2009, Piani_subchannel}. Uncovering an analogous operational interpretation of these quantifiers in the context of work extraction would be an exciting direction for future research. 

\section{Acknowledgement}
We acknowledge Chung-Yun Hsieh for his valuable comments. This material is based upon work supported by the U.S. Department of Energy, Office of Science, Accelerated Research in Quantum Computing, Fundamental Algorithmic Research toward Quantum Utility (FAR-Qu) and Fundamental Algorithmic Research in Quantum Computing (FAR-QC). We also acknowledge support from the Beyond Moore’s Law project of the Advanced Simulation and Computing Program at LANL managed by Triad National Security, LLC, for the National Nuclear Security Administration of the U.S. DOE under contract 89233218CNA000001.

\bibliography{main}
\appendix
\onecolumngrid
\section{Revisiting the concept of unsteerability}\label{Unsteerability_appendix}
Let us briefly revisit the concept of unsteerability of quantum states. Consider a scenario where Alice and Bob share a bipartite quantum state \(\hat{\rho}_{AB}\). Alice can perform different measurements on her system, which need not necessarily be projective. For each of Alice's measurement settings \(x\) and corresponding outcomes \(a\), Bob is left with a conditional state \(\hat{\rho}_{a|x}\). Once the set of states \(\{\hat{\rho}_{a|x}\}_{a,x}\) is characterized, Bob can attempt to explain their appearance as follows: Bob assumes that his system is initially in a hidden state \( \hat\rho_{\lambda} \), determined by a hidden variable \( \lambda \) occurring with probability \( p(\lambda) \). When Alice performs a measurement and obtains a result, this provides Bob with additional information about the probabilities associated with these states. Thus, the updated state can be written as 
\begin{equation}\label{defn:unsteerability}
    \hat{\rho}_{a|x}=\frac{\rho_{a|x}}{p(a|x)} = \int d\lambda\; p(\lambda|a,x)\hat\rho_{\lambda}\quad{\Longleftrightarrow}\quad \rho_{a|x} = \int d\lambda\; p(\lambda)\; p(a|x,\lambda)\;\hat\rho_{\lambda}\quad \text{for any $a,x$}.
\end{equation}
The set of unnormalized states $\{\rho_{a|x}\}_{a,x}$ referred to as  \emph{assemblages}). An assemblage is said to be \emph{unsteerable} if they admit the decomposition in Eq.~\eqref{defn:unsteerability} for any $a,x$. The equivalence between these two expressions is easy to verify assuming \( x \) can be freely chosen and is independent of the parameter $\lambda$, meaning \( p(x, \lambda) = p(x) p(\lambda) \). We denote the set of all unsteerable assemblages as $\mathcal{L}$.

The state $\hat{\rho}_{AB}$ is then said to be unsteerable. Here, Bob does not need to assume that Alice has any direct control over his state. Instead, her measurements and results simply provide additional information about the distribution of the hidden states  $\hat\rho_{\lambda}$. 

 Quantum steering encapsulates the non-classical correlations that can be observed between the outcomes of measurements applied on half of an entangled state and the resulting post-measurement states that are left with the other
party.  Operationally, a steering test can be viewed as an entanglement test in which one of the parties performs uncharacterized measurements. An uncharacterized measurement refers to a scenario where the measurement device is treated as a black box: given a classical input \( x \), it produces a classical output \( a \) with probability \( p(a|x) \), without any knowledge of the internal implementation or the specific measurement operators involved. In contrast, testing the entanglement of an arbitrary state typically requires that both parties perform well-characterized measurements, where the measurement operators are known to the corresponding parties. On the other extreme, in a Bell nonlocality test, both parties are assumed to use uncharacterized measurement devices. The relationship between entangled, steerable, and Bell nonlocal quantum states can be summarized as a strict hierarchy:
\begin{equation}
\text{Bell nonlocal} \subsetneq \text{Steerable} \subsetneq \text{Entangled}
\end{equation}
 Thus, quantum steering is a form of quantum inseparability
that lies in between the well-known notions of Bell nonlocality and entanglement. It is worth mentioning that steering may be demonstrated using any pure entangled state, and the same is
true of Bell nonlocality~\cite{Wiseman_steering}.

\section{Proof of Theorem \ref{LHS_theorem_upper_bound_work_extraction}}\label{Appendix_proof}
Equation~\eqref{Avg_amount_of_extracted_work} in the main text provides the average amount of extracted work. It holds that  
\begin{eqnarray}
    W\left(\{\hat\rho_{a|x}\},\{H_{a|x}\}\right) 
    &=& \frac{1}{|\mathcal{X}|}\sum_{a,x}p(a|x)\Big(-\Tr(H_{a|x}\hat\rho_{a|x})+\Tr(H_{a|x}\hat\gamma_{a|x})\Big)\label{exact_avg_work_gen} \\
    &=& -\frac{1}{|\mathcal{X}|}\sum_{a,x}\Tr(H_{a|x}\rho_{a|x})+\frac{1}{|\mathcal{X}|}\sum_{a,x}p(a|x)\Tr(H_{a|x}\hat\gamma_{a|x})\nonumber\\
    &=& -\frac{1}{|\mathcal{X}|}\sum_{a,x}\int\;d\lambda p(\lambda)p(a|x,\lambda)\Tr(H_{a|x}\hat\rho_{\lambda})+ \frac{1}{|\mathcal{X}|}\sum_{a,x}p(a|x)\Tr(H_{a|x}\hat\gamma_{a|x})=W_{\text{LHS}}\label{The_important_bound}.
\end{eqnarray}
To write the second equality, we use the fact that $p(a|x)\hat\rho_{a|x}=\rho_{a|x}$, whereas to write the third equality we use the definition of unsteerability from Eq.~\eqref{LHS_decomposition_assemblage_main}. Now, we consider the set of Hamiltonians $H_{a|x}$ from Eq. \eqref{MUB_Hamiltonian_main} in the main text, i.e.,
\begin{equation}\label{MUB_Hamiltonian}
    H_{a|x} =-\omega \ketbra{\phi^a_x}, \quad\text{such that}\; x\in\{0,1,\ldots,n-1\}\quad \text{and}\;a\in\{0,\ldots,d-1\},
\end{equation}
where $\left\{\ketbra{\phi^a_x}\right\}$ forms MUBs. For $H_{a|x}$ in Eq.~\eqref{MUB_Hamiltonian}, we have
\begin{eqnarray}\label{partition_fn}
   \Tr\left(H_{a|x}\hat\gamma_{a|x}\right)=-\omega\left(\frac{ e^{\beta\omega}}{e^{\beta\omega}+d-1}\right).
\end{eqnarray}
 We can provide an upper bound on the average amount of extracted work starting from Eq. \eqref{The_important_bound} as follows:
\begin{eqnarray}
    W_{\text{LHS}}&=&  -\frac{1}{n}\sum_{x=0}^{n-1}\sum_{a=0}^{d-1}\int\;d\lambda p(\lambda)p(a|x,\lambda)\Tr(H_{a|x}\hat\rho_{\lambda})+\frac{1}{ n}\sum_{x=0}^{d}\sum_{a=0}^{d-1}p(a|x)\Tr\left(H_{a|x}\hat\gamma_{a|x}\right)\nonumber\\
    &=& \frac{\omega }{n}\sum_{x=0}^{n-1}\sum_{a=0}^{d-1}\int\;d\lambda p(\lambda)p(a|x,\lambda)\Tr(\ketbra{\phi^a_x}\hat\rho_{\lambda})-\frac{\omega}{n}\sum_{x=0}^{n-1}\sum_{a=0}^{d-1}p(a|x)\left(\frac{ e^{\beta\omega}}{e^{\beta\omega}+d-1}\right)\nonumber\\
    &=& \frac{\omega }{n}\sum_{x=0}^{n-1}\sum_{a=0}^{d-1}\int\;d\lambda p(\lambda)p(a|x,\lambda)\Tr(\ketbra{\phi^a_x}\hat\rho_{\lambda})-\omega\left(\frac{ e^{\beta\omega}}{e^{\beta\omega}+d-1}\right)\nonumber\\
    &\leq& \frac{\omega}{n}\int d\lambda \;p(\lambda)\sum_{x=0}^{n-1}\max_{a}\Tr(\ketbra{\phi^a_x}\hat\rho_{\lambda})\underbrace{\sum_{a=0}^{d-1}p(a|x,\lambda)}_{=1}\;-\omega\left(\frac{ e^{\beta\omega}}{e^{\beta\omega}+d-1}\right)\nonumber\\
    &=& \omega \int d\lambda \;p(\lambda)\left[\frac{1}{n}\sum_{x=0}^{n-1}\Tr(\ketbra{\phi^{a_{\text{max}}}_x}\hat\rho_{\lambda})\right]\;-\omega\left(\frac{ e^{\beta\omega}}{e^{\beta\omega}+d-1}\right)\nonumber\quad\\
    &\leq& \omega\underbrace{\int d\lambda \;p(\lambda)}_{=1}\frac{1}{d}\left(1+\frac{d-1}{\sqrt{n}}\right)-\omega\left(\frac{ e^{\beta\omega}}{e^{\beta\omega}+d-1}\right)\nonumber\\
    &=& \frac{\omega}{d}\left(1+\frac{d-1}{\sqrt{n}}\right)-\omega\left(\frac{e^{\beta\omega}}{e^{\beta\omega}+d-1}\right)=W_{\text{classical}}\label{appendix_steering_bd},
\end{eqnarray}
where we define $\Tr(\ketbra{\phi^{a_{\text{max}}}_x}\hat\rho_{\lambda}) :=\max_{a}\Tr(\ketbra{\phi^a_x}\hat\rho_{\lambda})$. To write the second equality, we substitute $H_{a|x}$ from Eq. \eqref{MUB_Hamiltonian} and substitute the expression of $\Tr(H_{a|x}\hat\gamma_{a|x})$ from Eq. \eqref{partition_fn}. To write the final inequality we use a result from Ref.~\cite{Rastegin2015}, which we include next for completeness:
\begin{lem}[Proposition 1 of Ref. \cite{Rastegin2015}]
    Let $\left\{\ketbra{\phi_x^a}\right\}$ denotes the set of MUBs  in a $d$-dimensional Hilbert space indexed by $x$ such that $x\in\{0,\ldots,n-1\}$. For an arbitrary density matrix $\rho$, we have the following inequality:
\begin{equation}
    \frac{1}{n}\sum_{x=0}^{n-1} \Tr\left(\ketbra{\phi^a_x}\hat\rho\right)\leq \frac{1}{d}\left(1+\frac{d-1}{\sqrt{n}}\right)\quad\forall\; a\in\{1,\ldots,d\}.
\end{equation}
\end{lem}

\section{Derivation of Eq. \eqref{MUB_assemblage_main}}\label{Derivation_eq}

Recall that 
    $\hat{\rho}_{AB}=\ketbra{\psi_{AB}}$,
where
\begin{equation}
    \ket{\psi_{AB}}=\frac{1}{\sqrt{d}}\sum_{i=0}^{d-1}\ket{ii}.
\end{equation}
Suppose Bob randomly selects the input \( x \) by sampling from the set \( \mathcal{X} \). As Bob informs Alice of the input $x$, she performs a projective measurement on the basis
\begin{equation}
    \left\{\ketbra{\phi^{*c}_{x}}\right\}_{c=0}^{d-1},
\end{equation}
where 
\begin{equation}
    \ket{\phi^{*c}_{x}}=\sum_{i=0}^{d-1}\braket{i}{\phi^c_x}^{*}\ket{i}=\sum_{i=0}^{d-1}\braket{\phi^c_x}{i}\ket{i}.
\end{equation}
Alice obtains an outcome $a$. The above measurement yields the following state on Bob's side:
\begin{eqnarray}\label{MUB_assemblage}
    \hat\rho_{a|x}&=&\frac{\frac{1}{d}\Tr_{A}\left[\sum_{i=0}^{d-1}\sum_{j=0}^{d-1}\left(\ketbra{\phi^{*a}_{x}}\otimes \mathbb{I}\right)\left(\ketbra{i}{j}\otimes \ketbra{i}{j}\right)\right]}{\frac{1}{d}\Tr\left[\sum_{a=0}^{d-1}\sum_{\tilde{a}=0}^{d-1}\left(\ketbra{\phi^{*a}_{x}}\otimes \mathbb{I}\right)\left(\ketbra{i}{j}\otimes \ketbra{i}{j}\right)\right]}\nonumber\\
    &=& \left(\sum_{i=0}^{d-1}\braket{\phi_x^{*a}}{i}\ket{i}\right)\left(\sum_{j=0}^{d-1}\braket{j}{\phi_x^{*a}}\bra{j}\right)=\left(\sum_{i=0}^{d-1}\braket{i}{\phi^a_x}\ket{i}\right)\left(\sum_{j=0}^{d-1}\braket{j}{\phi^a_x}^{*}\bra{j}\right)
    =\ketbra{\phi^a_x}.
\end{eqnarray}
This proves Eq.~\eqref{MUB_assemblage_main} in the main text.

\section{Work extraction using Hamiltonians composed of Pauli string}\label{Pauli_strings_Hamiltonian}

In this section, we aim to design a work extraction task that demonstrates quantum advantage with Hamiltonian constructed from Pauli observables. We consider a scenario in which Bob can quench his system to Hamiltonians of the form  
\begin{equation}\label{Pauli_Hamiltonian_main}
    H_{a|x} := -\frac{\omega}{2}\left(\mathbb{I} + a A_x\right), \quad \text{with } a \in \{+1, -1\},
\end{equation}
where each \( A_x \), for \( x \in \{1, \cdots, n\} \), is a distinct, nontrivial Pauli string that are mutually anticommuting, and $\mathbb{I}$ is $2^{n}\times 2^{n}$ identity matrix. They are acting on the Hilbert space \( \mathcal{B}(\mathbb{C}^{2^n}) \).  In other words, we consider
\begin{equation}\label{Pauli_Hamiltonian_main3}
    A_x = \bigotimes_{i=1}^{n} A_x^{(i)}, \quad \text{with } A_x^{(i)} \in \{I, X, Y, Z\}\quad\text{and}\quad x\in\{1,\ldots,n\},
\end{equation}
which satisfy the following anticommutation relation:  
\begin{equation}\label{Eq:Anticommutation}
    A_x A_y + A_y A_x = 2\delta_{x,y}\, \mathbb{I}.
\end{equation}
Here, $I,X,Y,Z$ are the Pauli matrices. For example, one can take $A_x=Z^{\otimes (x-1)}\otimes X\otimes I^{\otimes (n-x)}$ i.e.,
\begin{eqnarray}
    A_1&=&X\otimes I\otimes I\otimes\cdots\otimes I,\nonumber\\
    A_2&=&Z\otimes X\otimes I\otimes\cdots\otimes I,\nonumber\\
    &&\quad\quad\vdots\quad\quad\vdots\quad\quad\nonumber\\
    A_n&=&Z\otimes Z\otimes Z\otimes\cdots\otimes X.
\end{eqnarray}
As $x\in\{1,\ldots, n\}$ which give $|\mathcal{X}|=n$. We now evaluate the maximum extractable work under the assumption that the correlations in the shared quantum state between Alice and Bob are unsteerable.
\subsection{Extractable work for unsteerable correlations}
\begin{thm}\label{LHS_theorem_upper_bound_work_extraction_2}
    Consider a work extraction task in which Bob is allowed to quench his system to the Hamiltonians \( H_{a|x} \) as defined in Eq.~\eqref{Pauli_Hamiltonian_main}, and the state \( \hat{\rho}_{a|x} \) admits a local-hidden-state (LHS)  decomposition as given in Eq. \eqref{defn:unsteerability}.
     Then, the average extractable work, denoted by \(  W_\mathrm{LHS} \), is bounded by
    \begin{eqnarray}\label{WLHS_Bound2}
         W_\mathrm{LHS} &\leq& \omega\left(\frac{1}{2\sqrt{n}}+\frac{1}{2}\right)-\frac{\omega}{1+e^{-\beta\omega}} \nonumber\\
        &=& \omega\left(\frac{1}{2\sqrt{\log_2 d}}+\frac{1}{2}\right)-\frac{\omega}{1+e^{-\beta\omega}}\eqqcolon W_{\mathrm{classical}}.
    \end{eqnarray}
\end{thm}
Having \( d = 2^n \) allows us to write equality in Eq.~\eqref{WLHS_Bound2}. Similar to the previous example, the second term on the right-hand side of the inequality in Eq.~\eqref{WLHS_Bound2} corresponds to the average energy of the Gibbs state with respect to the Hamiltonian defined in Eq.~\eqref{Pauli_Hamiltonian_main}. The proof of Theorem~\ref{LHS_theorem_upper_bound_work_extraction_2} is provided in Appendix~\ref{Appendix_proof_2}.

Similarly to the previous case, the inequality in Eq.~\eqref{WLHS_Bound2} can be employed as a steerability witness. In particular, whenever the extracted amount of work exceeds the classical threshold \(  W_{\text{classical}} \), we can conclude that the shared quantum state \( \hat\rho_{AB} \) between Alice and Bob possesses steerable correlations.

\subsection{Extractable Work in the Quantum Scenario}

In this section, we analyze the average extracted work defined in Eq.~\eqref{Avg_amount_of_extracted_work} within the quantum regime, where the assemblage \(\hat{\rho}_{a|x}\) is not constrained to admit a local-hidden-state (LHS) decomposition as in Eq.~\eqref{defn:unsteerability}.

\begin{thm}\label{Quantum_theorem_upper_bound_work_extraction2}
Consider the work extraction task in which Bob is allowed to quench his system to the Hamiltonians \(  H_{a|x} \) specified in Eq.~\eqref{Pauli_Hamiltonian_main}. The maximum achievable average extracted work in this quantum setting is  
\begin{equation}\label{Wquantum_Bound2}  
     W_{\mathrm{quantum}} \coloneqq \omega -\frac{\omega}{1+e^{-\beta\omega}}.  
\end{equation}
\end{thm}
\begin{proof}
We derive the upper bound on extractable work using an argument similar to the previous case. Recall that, for any quantum state, the average extractable work defined as:
\begin{eqnarray}\label{Avg_amount_of_extracted_work}
    &&W\left(\{\hat\rho_{a|x}\},\{ H_{a|x}\}\right) =  \frac{1}{|\mathcal{X}|}\sum_{a, x} p(a|x) \left[-\Tr\left( H_{a|x} \hat{\rho}_{a|x} \right) + \Tr( H_{a|x} \hat\gamma_{a|x})\right].
\end{eqnarray}
Similarly as before, average extracted amount of work can be upper bounded as follows  
\begin{equation}
\label{upper_bound_main_work_main2}
    W\left(\{\hat\rho_{a|x}\},\{ H_{a|x}\}\right) \leq \frac{1}{|\mathcal{X}|} \sum_{a, x} p(a|x) \left[ -\sigma_{\min}( H_{a|x}) + \Tr\left( H_{a|x} \hat{\gamma}_{a|x} \right) \right],
\end{equation}
where \( \sigma_{\min}(\cdot) \) denotes the smallest eigenvalue of the corresponding Hamiltonian. In the case of the Hamiltonians \(  H_{a|x} \) defined in Eq.~\eqref{Pauli_Hamiltonian_main}, we have \( \sigma_{\min}( H_{a|x}) = -\omega \). Substituting this into Eq.~\eqref{upper_bound_main_work_main2}, and noting that
\begin{equation}\label{Gibbs_2_energy}
    \frac{1}{|\mathcal{X}|}\sum_{a,x}p(a|x)\Tr( H_{a|x}\gamma_{a|x})=-\frac{\omega}{1+e^{-\beta\omega}},
\end{equation}
we obtain the upper bound:
\begin{equation}
\label{generic_upper_bound_optimality2}
    W \leq \omega - \frac{\omega}{1 + e^{-\beta \omega}}.
\end{equation}
We have provided the derivation of Eq. \eqref{Gibbs_2_energy} in appendix \ref{Gibbs_energy}. 

    Next, we show a protocol that attains such a bound. Assume the initial state $\hat{\rho}_{AB}$ of the combined system of Alice and Bob is maximally entangled state of the form
    \begin{equation}
        \hat\rho_{AB}=\frac{1}{2^n}\sum_{i=0}^{2^n-1}\sum_{j=0}^{2^n-1}\ketbra{i}{j}\otimes\ketbra{i}{j}.
    \end{equation}
    Suppose Bob selects an input \( x \) uniformly at random from the set \(\mathcal{X}\) in each run. Upon receiving \( x \), Alice performs a dichotomic projective measurement in the orthonormal basis with corresponding projector
    \begin{eqnarray}        \Pi^{(x)}_{a}:=\frac{1}{2}\left(\mathbb{I}+aA_x^{T}\right)\quad\text{where}\quad a\in\{+1,-1\},
    \end{eqnarray}
    where $``T"$ in the superscript denotes transpose. Then the post-measurement state 
    \begin{equation}\label{Pauli_assemblage_main}
        \hat{\rho}_{a|x} = \frac{1}{2^n}\left(\mathbb{I}+aA_x\right),
    \end{equation}
    is obtained with probability $1/2$. We derive Eq. \eqref{Pauli_assemblage_main} in appendix \ref{Post_mmt_Pauli}. Substituting $\hat{\tilde\rho}_{a|x}$ from Eq. \eqref{Pauli_assemblage_main} and $ H_{a|x}=-\frac{\omega}{2}\left(\mathbb{I} + a A_x\right)$ in the extractable work in Eq. \eqref{Avg_amount_of_extracted_work}, we obtain that the extractable work $ W_{\text{quantum}}$ in the quantum scenario is
\begin{equation}\label{W_quantum_main2}
     W_{\text{quantum}}=\omega-\frac{\omega}{1+e^{-\beta\omega}}.
\end{equation}
\end{proof}

\subsection{Quantum thermodynamic advantage in work extraction}

Using Theorems~\ref{LHS_theorem_upper_bound_work_extraction_2} and~\ref{Quantum_theorem_upper_bound_work_extraction2}, the ratio between the extractable work in the quantum and unsteerable (LHS) scenarios satisfies  
\begin{equation}\label{definition2}
    \frac{ W_{\text{quantum}}}{ W_{\text{LHS}}} \geq \frac{ W_{\text{quantum}}}{ W_{\text{classical}}} = \frac{\omega - \frac{\omega}{1 + e^{-\beta \omega}}}{\omega\left( \frac{1}{2\sqrt{\log_2 d}} + \frac{1}{2} \right) - \frac{\omega}{1 + e^{-\beta \omega}}} \coloneqq  \xi.
\end{equation}

The quantity \(  \xi \) captures the quantum thermodynamic advantage in the work extraction task. As in the previous scenario, a quantum advantage is achieved whenever \(  \xi > 1 \), that is, when the extractable work in the quantum case exceeds that in the classical (or LHS) scenario: \(  W_{\text{quantum}} >  W_{\text{classical}} \geq  W_{\text{LHS}} \).

From Eq.~\eqref{definition2}, a sufficient condition for observing this advantage is
\begin{equation}\label{omegadsqrtn2}
    \omega > \omega\left( \frac{1}{2\sqrt{\log_2 d}} + \frac{1}{2} \right) \quad \Longleftrightarrow \quad 1 > \frac{1}{\sqrt{\log_2 d}},
\end{equation}
which holds for all dimensions \( d > 2 \). Consequently, Eq.~\eqref{definition2} implies that the ratio \(  \xi \sim \mathcal{O}(\sqrt{\log_2 d}) \), with the quantum extractable work scaling as \(  W_{\text{quantum}} \sim \mathcal{O}(1) \), and its classical counterpart scaling as \(  W_{\text{classical}} \sim \mathcal{O}(1/\sqrt{\log_2 d}) \), consistent with Eqs.~\eqref{W_quantum_main2} and~\eqref{WLHS_Bound2}.
\section{Proof of the theorem \ref{LHS_theorem_upper_bound_work_extraction_2} of Supplementary material}\label{Appendix_proof_2}
Assume that Bob can quench his system to Hamiltonians
\begin{equation}\label{Clifford_Hamiltonian}
     H_{-1|x} = -\frac{\omega}{2} (\mathbb{I}+A_x)\quad;\quad  H_{+1|x} = -\frac{\omega}{2} (\mathbb{I}-A_x), \quad\text{for}\quad x=\{1,\ldots,n\}, 
\end{equation}
where $A_x$ are non-trivial (i.e., $A_x\neq I^{\otimes n}$), distinct Pauli strings (i.e $A_x\neq A_y$ if $x\neq y$) obeying Eq. \eqref{Eq:Anticommutation}. As $x\in\{1,\ldots, n\}$ which implies $|\mathcal{X}|=n$. Let us now bound the extracted amount of work for unsteerable correlations starting from Eq. \eqref{Avg_amount_of_extracted_work}:
   \begin{eqnarray}
    W&=&\frac{1}{|\mathcal{X}|}\sum_{a,x}p(a|x)\Big(-\Tr( H_{a|x}\hat\rho_{a|x})+\Tr(H_{a|x}\hat\gamma_{a|x})\Big)\label{exact_avg_work_gen_2}\\
    &=&\frac{1}{n}\sum_{a,x}p(a|x)\Big(-\Tr( H_{a|x}\hat\rho_{a|x})+\Tr(H_{a|x}\hat\gamma_{a|x})\Big)\\
    &=& -\frac{1}{n}\sum_{a,x}\Tr( H_{a|x}\rho_{a|x})+\frac{1}{n}\sum_{a,x}p(a|x)\Tr( H_{a|x}\hat\gamma_{a|x})\nonumber\\
    &=& -\frac{1}{n}\sum_{a,x}\int\;d\lambda p(\lambda)p(a|x,\lambda)\Tr( H_{a|x}\hat\rho_{\lambda})+ \frac{1}{n}\sum_{a,x}p(a|x)\Tr( H_{a|x}\hat\gamma_{a|x})={W}_{\text{LHS}}\label{The_important_bound_2},
\end{eqnarray}
where to write third equality, we use $\rho_{a|x}=p(a|x)\hat\rho_{a|x}$, and to write fourth equality we use Eq. \eqref{defn:unsteerability}. Using the Hamiltonian from Eq. \eqref{Clifford_Hamiltonian}, we obtain
\begin{eqnarray}\label{GE_av_GE}
     \frac{1}{n}\sum_{a,x}p(a|x)\Tr( H_{a|x}\hat\gamma_{a|x}) = \frac{1}{n}\sum_{a,x}p(a|x)\left(-\frac{\omega }{1+e^{-\beta\omega}}\right)=\left(-\frac{\omega }{1+e^{-\beta\omega}}\right).
\end{eqnarray}
The derivation of the above Eq. \eqref{GE_av_GE} is given in section \ref{Gibbs_energy}.
Next, we upper bound the first term of Eq.~\eqref{The_important_bound_2} as follows:
\begin{eqnarray}
    -\frac{1}{n}\sum_{a,x}\int\;d\lambda p(\lambda)p(a|x,\lambda)\Tr(H_{a|x}\hat\rho_{\lambda})&=&-\frac{1}{n}\Tr\left(\int\;d\lambda p(\lambda)\hat\rho_{\lambda}\sum_{a,x} p(a|x,\lambda) H_{a|x}\right)\nonumber\\
    &=& -\frac{1}{n}\Tr\left(\int\;d\lambda p(\lambda)\hat\rho_{\lambda}\left(\sum_{x} p(+1|x,\lambda) H_{+1|x}+\sum_{x}p(-1|x,\lambda)H_{-1|x}\right)\right)\nonumber\\
    &=&\frac{\omega}{n}\Tr\left(\int\;d\lambda p(\lambda)\hat\rho_{\lambda}\left(\sum_{x} p(+1|x,\lambda)\left(\frac{\mathbb{I}+A_x}{2}\right)+\sum_{x}p(-1|x,\lambda)\left(\frac{\mathbb{I}-A_x}{2}\right)\right)\right)\nonumber\\
    &=&\frac{\omega}{n}\Tr\left(\int\;d\lambda p(\lambda)\hat\rho_{\lambda}\left(\sum_{x} p(+1|x,\lambda)\left(\frac{A_x}{2}\right)+\sum_{x}p(-1|x,\lambda)\left(\frac{-A_x}{2}\right)+\sum_{x}\frac{\mathbb{I}}{2}\right)\right)\nonumber\\
    &=&\frac{\omega}{n}\Tr\left(\int\;d\lambda p(\lambda)\hat\rho_{\lambda}\left(\sum_{x} \left\{p(+1|x,\lambda)-p(-1|x,\lambda)\right\}\left(\frac{A_x}{2}\right)+\frac{n\mathbb{I}}{2}\right)\right)\nonumber\\
    &=&\frac{\omega}{n}\Tr\left(\int\;d\lambda p(\lambda)\hat\rho_{\lambda}\left(\sum_{x} \left\{p(+1|x,\lambda)-1+p(+1|x,\lambda)\right\}\left(\frac{A_x}{2}\right)+\frac{n\mathbb{I}}{2}\right)\right)\nonumber\\
    &=&\frac{\omega}{n}\Tr\left(\int\;d\lambda p(\lambda)\hat\rho_{\lambda}\left(2\sum_{x} \left\{p(+1|x,\lambda)-\frac{1}{2}\right\}\left(\frac{A_x}{2}\right)\right)\right)+\frac{\omega}{2}\nonumber\\
    &=&\frac{\omega}{n}\Tr\left(\int\;d\lambda p(\lambda)\hat\rho_{\lambda}\left(\sum_{x} \underbrace{\left\{p(+1|x,\lambda)-\frac{1}{2}\right\}}_{:=b_{x,\lambda}}A_x\right)\right)+\frac{\omega}{2}\nonumber\\
    &=& \frac{\omega}{n}\Tr\left(\int\;d\lambda p(\lambda)\hat\rho_{\lambda}\left(\sum_{x} b_{x,\lambda}A_x\right)\right)+\frac{\omega}{2}\nonumber\\&\leq& \frac{\omega}{n}\Tr\left(\int\;d\lambda p(\lambda)\hat\rho_{\lambda}\right)\Bigg\|\left(\sum_{x} b_{x,\lambda}A_x\right)\Bigg\|_{\infty}+\frac{\omega}{2}\label{Eq:Tr_ineq}\\
    &=& \frac{\omega}{n}\Bigg\|\left(\sum_{x} b_{x,\lambda}A_x\right)\Bigg\|_{\infty}+\frac{\omega}{2},\label{Eq:Final_form_Clifford}
\end{eqnarray}
where $\|\cdot\|_{\infty}$ denotes the operator norm, i.e., the largest eigenvalue of the operator $(\cdot)$. To write the inequality in Eq. \eqref{Eq:Tr_ineq}, we use $\Tr(AB)\leq \|A\|_{\infty}\Tr(B)$.  Note that $-\frac{1}{2}\leq b_{x,\lambda}=p(+1|x,\lambda)-\frac{1}{2}\leq \frac{1}{2}$.  For any hermitian operator $R$,
we have $\|R^2\|_{\infty}=\|R\|_{\infty}^2$, which allows us to write: 
\begin{equation}\label{ineq:app:last}
    \Bigg\|\left(\sum_{x} b_{x,\lambda}A_x\right)\Bigg\|_{\infty}^2 = \Bigg\|\left(\sum_{x} b_{x,\lambda}A_x\right)\left(\sum_{y} b_{y,\lambda}A_y\right)\Bigg\|_{\infty}= \left\| \sum_{x}b_{x,\lambda}^2 A_x^2 + \underbrace{\sum_{x \neq y} b_{x,\lambda} b_{y,\lambda} A_x A_y}_{=0}\right\|_\infty = \left\| \sum_x b_{x,\lambda}^2 \mathbb{I} \right\|_\infty \leq \frac{n}{4},
\end{equation}
where to write the final equality, we use the fact that $(A_xA_y+A_yA_x)=0$ when $x\neq y$. To write the second inequality we use the fact $b_{x,\lambda}\leq \frac{1}{2}$. From inequality Eq. \eqref{ineq:app:last} we can write
\begin{equation}
    \frac{\omega}{n}\Bigg\|\left(\sum_{x} b_{x,\lambda}A_x\right)\Bigg\|_{\infty}+\frac{\omega}{2} \leq \frac{\omega}{n}\sqrt{\frac{n}{4}}+\frac{\omega}{2} = \omega\left(\frac{1}{2\sqrt{n}}+\frac{1}{2}\right)=\omega\left(\frac{1}{2\sqrt{\log_2 d}}+\frac{1}{2}\right),
\end{equation}
where to write the final equality we use $d=2^n$. Thus, we can bound the extracted amount of work $W$ starting from Eq. \eqref{The_important_bound_2} as follows:

\begin{equation}
     W_{\text{LHS}} \leq \omega\left(\frac{1}{2\sqrt{n}}+\frac{1}{2}\right)-\frac{\omega}{1+e^{-\beta\omega}}=\omega\left(\frac{1}{2\sqrt{\log d}}+\frac{1}{2}\right)-\frac{\omega}{1+e^{-\beta\omega}}= W_{\text{Classical}}.
\end{equation}

\section{Calculation of average energy of the Gibbs state wrt. Hamiltonians constructed from Pauli strings}\label{Gibbs_energy}
We consider the Hamiltonian:
\begin{equation}
 H_{a|x} = -\frac{\omega}{2}(\mathbb{I} + a A_x)\quad\text{where}\quad a\in\{+1,-1\}\quad\text{and}\quad x\in\{1,\ldots,n\}
\end{equation}
with \( \omega > 0 \). Here \( A_x \) is a nontrivial Pauli strings, so \( A_x^2 = \mathbb{I} \) and \( \operatorname{Tr}(A_x) = 0 \). Since for any $x\in\{1,\ldots,n\}\;\;$, \( A_x \) has eigenvalues \( \pm 1 \). We define projectors:
\begin{equation}\label{Eq:Pauli_Proj}
\Pi_{a}^{(x)} := \frac{1}{2}(\mathbb{I} + a A_x),
\end{equation}
so that $ H_{a|x} = -\omega \Pi_{a}^{(x)}$. Then the Gibbs state can be written  as:
\begin{equation}
\hat{\gamma}_{a|x} = \frac{e^{-\beta  H_{a|x}}}{\operatorname{Tr}(e^{-\beta  H_{a|x}})}.
\end{equation}
Using the projector form of \(  H_{a|x} = -\omega \Pi_{a}^{(x)} \), we have $e^{-\beta  H_{a|x}} = e^{\beta\omega \Pi_{a}^{(x)}} = \Pi_{a}^{(x)} e^{\beta \omega} + \left(\mathbb{I}-\Pi_{a}^{(x)}\right)$. Hence, the Gibbs state becomes:
\begin{equation}\label{hatgamma}
\hat{\gamma}_{a|x} = \frac{\Pi_{a}^{(x)} e^{\beta \omega} + \left(\mathbb{I}-\Pi_{a}^{(x)}\right)}{\Tr(\Pi_{a}^{(x)}) e^{\beta \omega} + \Tr\left(\mathbb{I}-\Pi_{a}^{(x)}\right)}.
\end{equation}
We now compute the average energy:
\begin{equation}
\operatorname{Tr}( H_{a|x} \hat{\gamma}_{a|x}) = -\omega \Tr(\Pi_{a}^{(x)} \hat{\gamma}_{a|x}).
\end{equation}
Upon substituting the expression for \( \hat{\gamma}_{a|x} \) from Eq. \eqref{hatgamma} we obtain:
\begin{equation}
\operatorname{Tr}(\Pi_{a}^{(x)} \hat{\gamma}_{a|x}) = \frac{e^{\beta \omega} \operatorname{Tr}(\Pi_{a}^{(x)})}{\Tr(\Pi_{a}^{(x)}) e^{\beta \omega} + \Tr\left(\mathbb{I}-\Pi_{a}^{(x)}\right)}.
\end{equation}
Since \( A_x \) is traceless, \( \operatorname{Tr}(\Pi_{a}^{(x)}) = \operatorname{Tr}(\mathbb{I}-\Pi_{a}^{(x)}) = d/2 \), where \( d = 2^n \). Thus:
\begin{equation}
\operatorname{Tr}(\Pi_a^{(x)} \hat{\gamma}_{a|x}) = \frac{e^{\beta \omega} \cdot \frac{d}{2}}{e^{\beta \omega} \cdot \frac{d}{2} + \frac{d}{2}} = \frac{e^{\beta \omega}}{1 + e^{\beta \omega}}.
\end{equation}
Therefore, the average energy is:
\begin{equation}
\operatorname{Tr}(H_{a|x} \hat{\gamma}_{a|x}) = -\omega \cdot \frac{e^{\beta \omega}}{1 + e^{\beta \omega}} = -\frac{\omega}{1 + e^{-\beta \omega}},
\end{equation}
which allows us to write:
\begin{eqnarray}
     \frac{1}{n}\sum_{a,x}p(a|x)\Tr( H_{a|x}\hat\gamma_{a|x}) = \frac{1}{n}\sum_{a,x}p(a|x)\left(-\frac{\omega }{1+e^{-\beta\omega}}\right)=\left(-\frac{\omega }{1+e^{-\beta\omega}}\right).
\end{eqnarray}
\section{Derivation of Post measurement state }\label{Post_mmt_Pauli}
We consider the bipartite maximally entangled state on \( \mathbb{C}^{2^n} \otimes \mathbb{C}^{2^n} \),
\begin{equation}
    \hat{\rho}_{AB} = |\Phi^+\rangle\langle\Phi^+|, \quad \text{where } |\Phi^+\rangle = \frac{1}{\sqrt{2^n}} \sum_{i=0}^{2^n - 1} |i\rangle \otimes |i\rangle.
\end{equation}
Suppose a projective measurement is performed on subsystem \( A \) described by  the projector $R^{(x)}_a$ where 
\begin{equation}
    R^{(x)}_a = \frac{1}{2}\left( \mathbb{I} + a A^{T}_x \right), \quad a \in \{+1, -1\},
\end{equation}
where \( A_x \) is a nontrivial \( n \)-qubit Pauli strings, and $T$ in superscript denotes transposition. It is easy to see that $R^{(x)}_a$ is same as the transposition of $\Pi^{(x)}_a$ from Eq. \eqref{Eq:Pauli_Proj}. Then, the post-measurement state on subsystem \( B \) is given by
\begin{equation}
    \hat{\rho}^{(a|x)}_B = \frac{\operatorname{Tr}_A\left[ \left( R^{(x)}_a \otimes \mathbb{I} \right) \hat{\rho}_{AB} \right]}{\operatorname{Tr}\left[ \left( R^{(x)}_a \otimes \mathbb{I} \right) \hat{\rho}_{AB} \right]}.
\end{equation}
Using the following equality for maximally entangled states and linear operator $M$, specifically
\begin{equation}
    \operatorname{Tr}_A\left[ \left( M^{T} \otimes \mathbb{I} \right) |\Phi^+\rangle\langle\Phi^+| \right] = \frac{1}{2^n} M,
\end{equation}
we obtain
\begin{equation}
    \hat{\rho}^{(a|x)}_B = \frac{1}{2^{n}} \left( \Pi^{(x)}_a \right) 
    = \frac{1}{2^{n}} \left( \mathbb{I} + a A_x \right).
\end{equation}
This allow us to obtained the post-measurement state given in Eq. \eqref{Pauli_assemblage_main}.
\end{document}